\documentclass[twocolumn]{article}

\usepackage{amsmath}
\usepackage{eufrak}
\usepackage{units}
\usepackage{graphicx}
\usepackage{latexsym}
\usepackage{paralist}
\usepackage{circle}
\usepackage{amssymb}
\usepackage{changepage}
\newtheorem{definition}{Definition}
\newtheorem{proposition}{Proposition}
\newtheorem{proof}{Proof}
\usepackage{amsmath}
\usepackage{latexsym}
\usepackage{algorithm}
\DeclareMathOperator*{\argmax}{\arg\!\max}
\usepackage{tikz}

\usetikzlibrary{arrows,automata}
\tikzset{
    position/.style args={#1:#2 from #3}{
        at=(#3.#1), anchor=#1+180, shift=(#1:#2)
    }
}
\usepackage{graphicx}
\usepackage{subcaption}
\usepackage[colorinlistoftodos]{todonotes}
\usepackage{authblk}

\newcommand{\specialcell}[2][c]{%
  \begin{tabular}[#1]{@{}c@{}}#2\end{tabular}}

\title{Vocabulary Alignment in Openly Specified Interactions}

\author[1,2]{Paula Chocron}
\author[1]{Marco Schorlemmer}
\affil[1]{Artificial Intelligence Research Institute (IIIA-CSIC), Bellaterra, Catalonia, Spain}
\affil[2]{Universitat Aut\`onoma de Barcelona, Bellaterra, Catalonia, Spain}

\date{}

\begin{document}

\maketitle
\begin{abstract}
The problem of achieving common understanding between agents that use different vocabularies has been mainly addressed by designing techniques that explicitly negotiate mappings between their vocabularies, requiring agents to share a meta-language. In this paper we consider the case of agents that use different vocabularies and have no meta-language in common, but share the knowledge of how to perform a task, given by the specification of an interaction protocol. For this situation, we present a framework that lets agents learn a vocabulary alignment from the  experience of interacting. Unlike previous work in this direction, we use open protocols that constrain possible actions instead of defining procedures, making our approach more general. We present two techniques that can be used either to learn an alignment from scratch or to repair an existent one, and we evaluate experimentally their performance. 
\end{abstract}

\section{Introduction}

Addressing the problem of vocabulary heterogeneity is necessary for the common understanding of agents that use different languages, and therefore crucial for the success of multi-agent systems that act jointly by communicating. This problem has been tackled several times in the past two decades, in general from two different perspectives. Some approaches \cite{Steels1998,diggelen2007ontology} consider the existence of external \emph{contextual} elements that all agents perceive in common, and explore how those can be used to explain the meaning of words. A second group of techniques \cite{Santos2016, Silva2005} (and also \cite{diggelen2007ontology}, where both approaches are combined) consider the situation, reasonable in agents that communicate remotely, in which this kind of context is not available. They do so by providing explicit ways of learning or agreeing on a common vocabulary (or alignment between vocabularies). These techniques require agents to share a common meta-language that they can use to discuss about the meaning of words and their alignments. The complex question of how to communicate with heterogeneous interlocutors when neither a physical context nor a meta-language are available remains practically unexplored. 

\begin{figure}
\begin{center}

\begin{tikzpicture}[>=stealth',shorten >=1pt,auto,node distance=2cm]
  \node[state, scale=0.45]          (s0)      {\footnotesize $0$};
  \node[state, scale=0.45, position=0:{2.5cm} from s0]              (s1) [right of=s0]  {\footnotesize $1$};
  \node[state, scale=0.45, position=0:{1.8cm} from s1]              (s2) [right of=s1] {\footnotesize $3$};
  \node[state, scale=0.45, position=+5:{2.4cm} from s1, accepting]              (s3) [above right of=s1] {\footnotesize $2$};
  \node[state, scale=0.45, position=0:{2.4cm} from s2]              (s5){\footnotesize $5$};
 \node[state, scale=0.45, position=-15:{2.5cm} from s5, accepting]              (s6)  {\footnotesize $6$};
 \node[state, scale=0.45, position=+10:{3.5cm} from s5, accepting]              (s7)  {\footnotesize $7$};

  \path[->] (s0)  edge node {{\tiny W: to drink?}} (s1)
          (s1)  edge node {{\tiny C: wine }} (s2)
          (s1) edge node {{\tiny C: beer}} (s3)
          (s2) edge node {{\tiny W: color?}} (s5)          
          (s5) edge node [below] {{\tiny C: red}} (s6)
          (s5) edge node [above]  {{\tiny C: white}} (s7);

\end{tikzpicture}
\begin{tikzpicture}[>=stealth',shorten >=1pt,auto,node distance=2cm]
  \node[state, scale=0.45]          (s0)      {\footnotesize $0$};
  \node[state, scale=0.45, position=0:{2.5cm} from s0]              (s1) [right of=s0]  {\footnotesize $1$};
  \node[state, scale=0.45, position=0:{1.8cm} from s1]              (s2) [right of=s1] {\footnotesize $3$};
  \node[state, scale=0.45, position=+5:{2.4cm} from s1, accepting]              (s3) [above right of=s1] {\footnotesize $2$};
  \node[state, scale=0.45, position=0:{2.4cm} from s2]              (s5){\footnotesize $5$};
 \node[state, scale=0.45, position=-15:{2.5cm} from s5, accepting]              (s6)  {\footnotesize $6$};
 \node[state, scale=0.45, position=+10:{3.5cm} from s5, accepting]              (s7)  {\footnotesize $7$};

  \path[->] (s0)  edge node {{\tiny W: da bere?}} (s1)
          (s1)  edge node {{\tiny C: vino }} (s2)
          (s1) edge node {{\tiny C: birra}} (s3)
          (s2) edge node {{\tiny W: tipo?}} (s5)          
          (s5) edge node [below] {{\tiny C: rosso}} (s6)
          (s5) edge node  [above] {{\tiny C: bianco}} (s7);

\end{tikzpicture}
\end{center}
\caption{Protocols for Waiter (W) and Customer (C)}
\label{fig:issa}
\end{figure}
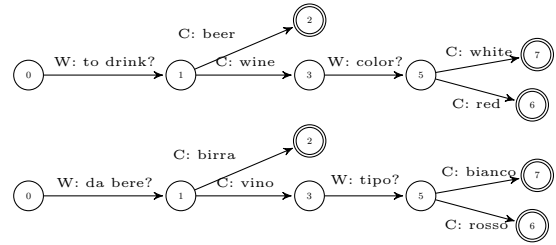


The work by Atencia and Schorlemmer \cite{Atencia2012} approaches this situation by considering a different kind of context, given by the interactions in which    agents engage. Agents are assumed to share the \emph{knowledge of how to perform a task}, or, more concretely, the specification of an interaction, given by a finite state automaton. For example, they consider the interaction in Figure \ref{fig:issa}, in which an English speaking customer and an Italian waiter communicate to order drinks. The authors show how agents can progressively learn which mappings lead to successful interactions from the experience of performing the task. After several interactions, agents converge to an alignment that they can use to always succeed at ordering and delivering drinks with that particular interlocutor. The interesting aspect of this idea is that the only shared element, namely the interaction specification, is already necessary to communicate. However, using finite state automata as specifications implies that agents need to agree on the exact order in which messages should be sent, which is unnecessarily restrictive. In addition, agents learn an alignment that is only useful for one task, and how it can be extrapolated to further interactions is not explored.

In this paper we present a general version of the \emph{interaction as context} approach to vocabulary alignment in multi-agent systems. Instead of automata, we consider agents that specify interactions with constrained-based \emph{open protocols} that define rules about what can be said instead of forcing one particular execution. In particular, we use ConDec protocols \cite{Pesic2006} that use linear temporal logic constraints. One ConDec protocol may not define completely the meaning of all words in a vocabulary, so we study how agents can learn mappings from performing different tasks. For this reason, the learning process is substantially different to the one in \cite{Atencia2012}, where a mapping that works in one task is considered correct. If enough experiences are available, agents that use the techniques we present converge to an alignment that is useful in the general case, even for future interactions that they do not yet know. Since  learning is done gradually as different interaction opportunities appear, agents can use what they learned early, even if they still do not know the alignment completely.

After presenting \emph{open interaction protocols}, we define a framework for interacting with  partners that use different vocabularies. We later present two different techniques to learn an alignment from the experience of interacting. The first relies only on analysing if a message is allowed or not in a particular moment, and it can be used in any constraint-based protocol.  The second technique uses the semantics of the protocols we present to improve the performance of the first one. Both methods can be used to learn alignments from scratch when there is no information, as well as to repair alignments obtained with other methods. We evaluate experimentally the techniques when used for both purposes, and show how different factors affect their performance. 

\section{Open Interaction Protocols}
\label{sec:prot}
The question of what should be expressed by an interaction protocol has been extensively discussed in the multi-agent systems community. Traditional approaches such as finite state automata and Petri Nets are simple to design and read, but also too rigid. More flexible alternatives constrain possible actions instead of defining a fixed procedure. \emph{Constraint Declarative Protocols} (commonly known as ConDec protocols) are an example of these approaches. ConDec protocols were first proposed by Pesic and van der Aalst \cite{Pesic2006} as a language to describe business protocols, and then used as specifications for agent interactions by Montali \cite{montali2010specification}, who presented an extension of ConDec and tools for its verification, and by Baldoni et al.\ \cite{Baldoni2010,Baldoni2015}, who integrated ConDec constrains and commitment protocols, a framework to specify interactions with social semantics first proposed by Singh \cite{singh2000asocial}. An important advantage of ConDec protocols is that they use linear temporal logic, a well-known logic for which many tools are available.

Linear temporal logic (LTL from now on) is a natural choice to express constraints about actions that occur in time. The syntax of LTL includes the one of propositional logic, and the additional temporal operators $\{\square, \lozenge, \Circle, U\}$, that can be applied to any LTL formula. LTL formulae are interpreted over \emph{paths} in Kripke structures, which are sequences of states associated to a truth-valuation of propositional variables. Temporal operators are interpreted in these paths as follows: $\square p$ means that $p$ must be true in the truth-valuation of all states, $\lozenge p$ means that $p$ must be true eventually, $\Circle p$ means that $p$ must be true in the next state, and $p \, U \, q$ means that $p$ must be true until $q$ is valid. A set of LTL formulae, called a \emph{theory}, is \emph{satisfiable} if there exists a path for which all the formulae are true. In that case, the path is a \emph{model} of the theory. The satisfiability problem in LTL is decidable, as well as the \emph{model checking} problem, which consists in deciding if a given path is a model of a theory. 

A ConDec protocol is a tuple $\langle M, C \rangle$, where $M$ is a set of action names and $C$ is a set of constraints about how actions can be performed. Constraints are LTL sentences renamed conveniently. We use a minimal version\footnote{ Since we are not interested in the usability of the protocols here, we do not include some constraints that work as syntactic sugar, such as $exactly(n,a)$, that can be replaced by including $existence(n,a)$ and $absence(n, a)$.} of the constraints introduced originally by Pesic and van der Aalst that we show in Table \ref{table:condec}, where $n \in \mathbb{N}$ and $a,b \in M$. These constraints are generally divided into three classes. Existential constraints ($existence$ and $absence$) predicate over the amount of times some action can be performed, relational constraints describe binary relations between two actions, and negation constraints (preceded by a `!' sign) are relations that do not hold. Given a set $M$ of actions, $Cons(M)$ is the set of all possible constraints over $M$. In a protocol $\langle M, C \rangle$, necessarily $C \subseteq Cons(M)$.

\begin{table}[h!]
\begin{center}
\bgroup
\def\arraystretch{1.3}
\begin{tabular}{ |c | c | }\hline
  Constraint & LTL meaning \\ \hline
  $existence(1,a)$ &  $\lozenge a$\\
  $existence(n+1,a)$ &  $ \lozenge (a \land \Circle existence(n,a))$\\
  $absence(0, a)$ & $\neg existence(1,a)$ \\
  $absence(n, a)$ & $\neg existence(n+1,a)$ \\
  $correlation(a,b)$ & $\lozenge a \implies \lozenge b$ \\
  $!correlation(a,b)$ & $\lozenge a \implies \neg \lozenge b$\\
  $response(a,b)$ & $\square (a \implies \lozenge b)$ \\
  $!response(a,b)$ & $\square (a \implies \neg \lozenge b)$\\
  $before(a,b)$ & $\neg b U a$ \\
  $!before(a,b)$ &  $\square (\lozenge b \implies \neg a$)\\
  $premise(a,b)$ & $\square (\Circle b \implies a)$ \\
  $!premise(a,b)$ & $\square (\Circle b \implies \neg a)$\\
  $imm\_after(a,b)$ & $\square (a \implies \Circle b)$ \\
  $!imm\_after(a,b)$ & $\square (a \implies \Circle \neg b)$ \\ \hline
\end{tabular}
\egroup
\end{center}
\caption{LTL definitions of constraints }
\label{table:condec}
\end{table}


\subsection{Open Protocols as Interaction Protocols}
In the rest of this section we present the technical notions that are necessary to use ConDec protocols as specifications of interactions between agents that may use different vocabularies. We first define \emph{interaction protocols}, that constrain the way in which agents can utter messages. We introduce the notion of \emph{bound} to capture the knowledge of agents about how many steps they have to finish the interaction. If this knowledge is not available it can be omitted and replaced for the fact that the interaction must finish in finite steps. 

\begin{definition}
  Given a set $A$ of agent IDs  and a propositional vocabulary $V$, an \emph{interaction protocol} is a tuple $\mathfrak{P} = \langle M, C, b  \rangle$, where the set of actions $M = A \times V$ is a set of messages formed by the ID of the sender agent and the term it utters, $C \subseteq Cons(M)$, and $b \in \mathbb{N}$ is the bound.
\end{definition}

The semantics of an interaction protocol is defined over interactions that represent a sequence of uttered messages.

\begin{definition}
	Given a set of messages $M$, an \emph{interaction} $i \in M^{*}$ is a finite sequence of messages $m \in M$. The length of an interaction ($len(i)$) and the append operation ($i  \, . \,  m$) are defined in the same way as for sequences. 
\end{definition}


An interaction $i$ can be encoded into a Kripke structure path by including one propositional variable for each message $m \in M$ and considering the following sequence of states with truth-valuations. In states with index $0 \leq j < len(i)$, let the propositional variable for the $j$-th message in $i$ be true, and all other propositional variables be false. For states after $len(i)$, let all propositional variables be false. With this construction, we can define the semantics of interaction protocols.

\begin{definition}
\label{def:model}
An interaction $i$ is a \emph{model} of an interaction protocol $\mathfrak{P} = \langle M, C, b  \rangle$ (noted $i \models \mathfrak{P}$) if $len(i) \leq bound$, and the Kripke path that encodes $i$ is a model of $C$. An interaction $i'$ is a \emph{partial model} (noted $i' \models_{p} \mathfrak{P}$)  of $\mathfrak{P}$ if it is a prefix of a model of $\mathfrak{P}$. 
\end{definition}


 Definition \ref{def:model} implies that checking satisfiability of an interaction protocol is equivalent to checking LTL satisfiability.

As already mentioned, we are interested in agents that use different vocabularies, but share the \emph{knowledge of how to perform a task}. In the rest of this section we define more precisely what that means. From now on, $V_1$ and $V_2$ are two possibly different vocabularies. Let us start by defining the notion of vocabulary alignment.

\begin{definition}
An \emph{alignment} is a function between vocabularies $\alpha : V_2 \rightarrow V_1$. Given a set of agents $A$, $\alpha$ can be extended homomorphically to a function between:
\begin{itemize}
  \item[-] messages $M_1 = A \times V_1$ and  $M_2 = A \times V_2$ ($ \alpha : M_2 \rightarrow M_1$)
  \item[-] constraints ($\alpha : Cons(M_2) \rightarrow Cons(M_1)$) and sets of constraints ($\alpha : 2^{Cons(M_2)} \rightarrow 2^{Cons(M_1)}$)
  \item[-] interactions ($\alpha : M_2^{*} \rightarrow M_1^{*}$) and sets of interactions ($\alpha : 2^{M_2^{*}} \rightarrow 2^{M_1^{*}}$)    
\end{itemize}

\end{definition}




To capture the idea of \emph{sharing the knowledge of how to perform a task}, we define notion of \emph{compatibility} between protocols, which consists simply on having the same models modulo an alignment. Given a protocol $\mathfrak{P} = \langle M, C, b \rangle$, we define the set of its models as $Int(\mathfrak{P}) = \{i \in M^{*} \text{ such that } i \models \mathfrak{P} \}$.

\begin{definition}
Two interaction protocols $\mathfrak{\mathfrak{P_1}} = \langle M_1, C_1, b \rangle$ and $\mathfrak{P_2} = \langle M_2, C_2, b \rangle$ with sets of messages $M_1 = A \times V_1$ and $M_2 = A \times V_2$ are \emph{compatible} if there exists an alignment $\alpha: V_2 \rightarrow V_1$ such that $$Int(\mathfrak{\mathfrak{P_1}}) = \alpha(Int(\mathfrak{P_2}))$$ If we know the alignment $\alpha$ for which the condition holds, we can say they are \emph{compatible under $\alpha$}.
\end{definition}

In this paper, for simplicity, we will restrict to using only bijective alignments over vocabularies of the same size.


\paragraph{An Example}
Let us extend the example of Figure \ref{fig:issa} by considering again a waiter W and  a customer C  and the \emph{ordering drinks} situation. The vocabulary of the customer is $V_C = $ \{\textsf{to drink, beer, wine, water, size, pint, half pint}\}, while the \emph{Waiter} uses $V_W = $ \{\textsf{da bere, birra, vino, acqua, tipo, media, piccola}\}. Consider the bijective alignment $\alpha : V_W \rightarrow V_C$ obtained by mapping each word in $V_W$ with the word in the same position in $V_C$. We consider the following protocols to specify the \emph{ordering drinks} interactions:

\begin{center}
$\mathfrak{P}_W =  \langle \{W,C\} \times V_W, $\\$ \{ existence(\langle W, \mathsf{da \;  bere} \rangle, 1, 1), $\\$ premise(\langle C,  \mathsf{birra} \rangle,  \langle W, \mathsf{da \; bere} \rangle), $\\$ premise(\langle C, \mathsf{vino} \rangle,  \langle W, \mathsf{da \; bere} \rangle), $\\$ premise(\langle C, \mathsf{acqua} \rangle, \langle W, \mathsf{da \; bere}\rangle), $\\$ !correlation(\langle C, \mathsf{birra} \rangle,  \langle W,\mathsf{vino} \rangle), $\\$ response(\langle C, \mathsf{birra} \rangle, \langle W, \mathsf{tipo} \rangle),$\\$ premise(\langle C, \mathsf{piccola} \rangle, \langle W, \mathsf{tipo} \rangle), $\\$ premise(\langle C, \mathsf{media}   \rangle, \mathsf{tipo} \rangle)\}, 5 \rangle$ \vspace{0.2cm}

$\mathfrak{P}_C = \langle \{W,C\} \times V_C, $\\$ \{ existence(\langle W, \mathsf{to \; drink} \rangle, 1, 1), $\\$ premise(\langle C,$\\$  \mathsf{beer} \rangle,  \langle W, \mathsf{to \; drink} \rangle), $\\$premise(\langle C, \mathsf{water} \rangle, \langle W, \mathsf{to \; drink} \rangle), $\\$ premise(\langle C, \mathsf{wine} \rangle, \langle W, \mathsf{to \; drink}), $\\$ response(\langle C, \mathsf{beer} \rangle,\langle W,  \mathsf{size} \rangle), $\\$ premise(\langle C, \mathsf{half pint} \rangle,\langle W,  \mathsf{size} \rangle), $\\$ premise(\langle C, \mathsf{pint}\rangle, \langle W, \mathsf{size}\rangle) \}, 5 \rangle$

\end{center}

The two protocols above are not compatible under any bijective alignment, since the customer protocol has as model an interaction in which the customer orders wine \emph{and} beer, while the waiter protocol has as models interactions in which the waiter only accepts one beverage. If $!correlation(\text{beer}, \text{wine}) $ is added to $\mathfrak{P}_C$, the resulting protocols are compatible, in particular under $\alpha$.

\section{Communicating with Heterogeneous Partners}
\label{sec:comm}

We focus on interactions between two agents $a_1$ and $a_2$ with vocabularies $V_1$ and $V_2$ respectively.
 During an interaction, an agent can send messages composed of words in its vocabulary  and receive others from its interlocutor, or finish the communication if certain conditions hold. We assume messages are never lost and always arrive in order, and more strongly, that each message is received before the following one is uttered. This requires that agents agree in who speaks at each time, although we do not force them to follow any particular turn-taking pattern.

Agents interact to perform some task together, that each of them specifies with an interaction protocol. We assume $a_1$ and $a_2$ agree on a set of tasks that they can perform and on \emph{how} they are performed: their respective protocols for each task are compatible. Our agents use their vocabulary consistently throughout different tasks, i.e., the protocols for all tasks are compatible under the same alignment. 

From now on, we will use sets of messages $M_1 = \{a_1,a_2\} \times V_1$ and $M_2 = \{a_1,a_2\} \times V_2$. We assume there exists a bijective alignment $\alpha: V_2 \rightarrow V_1$ such that, for each task that $a_1$ and $a_2$ can perform, they have protocols $\mathfrak{P}_1: \langle M_1, C_1, b \rangle$ and $\mathfrak{P}_2: \langle M_2, C_2, b \rangle$ respectively, and $\mathfrak{P}_1, \mathfrak{P}_2$ are compatible under $\alpha$. We assume that whenever agents interact, they use protocols that correspond to the same task.




We present a general approach to learn the alignment $\alpha$ from the experience of interacting to perform different tasks sequentially. The methods we present are used by one agent alone, and do not require that its interlocutor uses them as well. To explain the techniques, we adopt the perspective of agent $a_1$. 

Briefly, we propose to learn $\alpha$ by taking into account the coherence of messages. When $a_1$ receives a message, it learns from analysing which interpretations are allowed by the protocol and which are not. This information, however, is not definitive. An allowed word can be an incorrect interpretation for a received message, because protocols do not restrict completely all possible meanings. Moreover, a forbidden message can still be a correct interpretation. Since some rules express relations between messages, a previously misinterpreted word (by the agent or by its interlocutor) can make a correct mapping be impossible.

To handle this uncertainty, we propose a simple probabilistic learning technique. Agents maintain an \emph{interpretation distribution} over possible meanings for foreign words, and update their values with the experience of interacting. Formally, for each $v_2 \in V_2$ that $a_1$ knows about (because it received it in a message at some point) and for each $v_1 \in V_1$, agent $a_1$ has a weight $\omega(v_2, v_1)$ that represents its confidence in that $\alpha(v_2) = v_1$. Using the bijectivity, we will keep $\sum_{v \in V_1} \omega(v_2, v) = 1$, but we do not require the same in the other direction since the foreign vocabulary is unknown a priori.

The techniques we propose can incorporate existing alignments, in the same way as it is done in \cite{Chocron2016} for the case when protocols are automata. These alignments represent a priori knowledge about the foreign vocabulary, that can have been obtained from previous experience, from a matching tool, or with other techniques such as a syntactic similarity measure. Since these techniques are never fully correct, previous alignments should not be trusted completely, and the techniques we propose can work as methods for repairing them. A previous alignment for agent $a_1$ is a partial function $\mathcal{A} : V'_2 \times V_1 \rightarrow \mathbb{N}$, where $V'_2$ is a set of terms. Note that $\mathcal{A}$ is not necessarily defined over $V_2$, since the vocabulary that $a_2$ uses may be unknown a priori. However, the information it provides can be used even if $V'_2$ is a subset of $V_2$ or they overlap. We interpret the confidences always positively, meaning that a mapping with low confidence has still more confidence than one that does not exist. If $a_1$ has a previous alignment $\mathcal{A}$, it initializes the interpretation distribution using that information. They first assign raw values:

{\footnotesize
\begin{equation*}
\omega(v_2,v_1) =
\begin{cases}
    \mathcal{A}(v_2,v_1) & \mbox{if } (v_2,v_1) \in dom(\mathcal{A}) \vspace{0.2cm}\\ 
    0  & \mbox{otherwise}
\end{cases}
\end{equation*}
}

Then they normalize the values using an exponential method such as \emph{\textit{softmax}}. This is necessary to start the technique with the values for all mappings in the interval $(0,1)$, since the alignment can be wrong. If there is no previous alignment, they always initialize the interpretation values with an uniform distribution: $\omega(v_2,v_1) = \dfrac{1}{|V_1|}$ for all $v_1 \in V_1$.

We explain in detail the learning techniques that we propose in Section \ref{sec:learn}. In the rest of this section we focus on the dynamics of the interaction: how agents choose which messages to say, finish the interaction, and decide interpretations for the messages they receive.

\paragraph{Choosing messages to send}
The choice of messages to utter is internal of each agent and depends on its interests while interacting. We do not impose any restriction on this, besides respecting the constraints in the protocol.
Formally, $a_1$ can utter a message $v_1$ only if $\langle a_1, v_1 \rangle$ is a \emph{possible message} as defined below.

\begin{definition}
\label{def:pos}
	Consider an agent $a_1$ following protocol $\mathfrak{P_1} = \langle M_1, C_1, b \rangle$ and suppose interaction $i$ has happened so far (and $i \models_{p} \mathfrak{P_1}$). The \emph{possible messages} at that point are all messages $m \in M_1$ such that when they are performed the interaction remains a partial model, that is $i \, . \, m \models_{p} \mathfrak{P_1}$.
\end{definition}

\paragraph{Finishing the interaction}

 An interaction can finish in three situations. First, if there is a bound, reaching it implies the automatic end of the conversation. The interaction also finishes if an agent receives a message that has no possible interpretation. An agent can also finish the conversation whenever it wants if it considers that the interaction is \emph{successful}, i.e., if it  is a model of the protocol when it ends. In this case, it simply stops talking, producing a time-out that let the other agent realize the interaction finished.

\paragraph{Choosing interpretations}
When agent $a_1$ receives a message $v_2 \in V_2$ from $a_2$, it needs to interpret it in $V_1$ to be able to continue the interaction. To this aim, agents use the information in the interpretation distribution. Since agents assume there exists a bijective alignment, they always choose the same interpretation for a word during one interaction and they do not choose words that have been already chosen as interpretations.

Consider agent $a_1$ receives word $v_2$ after interaction $i$. Let $\mu : V_2 \rightarrow V_1$ be the \emph{mappings made} function, where $\mu(v_2) = v_1$ if and only if $v_1$ was chosen as a mapping for $v_2$ before in the same interaction. The domain $dom(\mu)$ are the $v_2 \in V_2$ that $a_1$ received since the current task started, and its image $img(\mu)$ are the $v_1 \in V_1$ that were chosen as mappings. The set $W$ of possible interpretations  for $v_2$ is the set of words $v_1$ in $V_1$ such that $\langle a_2, v_1 \rangle$ is a possible message after $i$, and such that $v_1 \not\in img(\mu)$.

If $rnd(S)$ is a function that chooses an element randomly in the set $S$, the agent will choose an interpretation as follows:

{\footnotesize
\begin{equation*}
\mu(v_2) =
\begin{cases}
     \mu(v_2) & \hspace{-1.8cm} \mbox{if } v_2 \in dom(\mu) \land \mu(v_2) \in W \vspace{0.2cm}\\ 
    rnd(\argmax_{v_1 \in W} \omega(v_2,v_1))  & \mbox{if }  v_2 \not\in dom(\mu) 
\end{cases}
\end{equation*}
}

If either $v_2 \in dom(\mu)$ but $\mu(v_2) \not\in W$, or $W = \emptyset$, the interaction is in one of the failure cases, and it finishes because $a_1$ stops talking.  


\section{Learning Alignments from Interactions}
\label{sec:learn}

 To learn an alignment from the experience of interacting, agents make the following \emph{well behaviour} assumptions about their interlocutor. Essentially, these assumptions imply that the dynamics of the interaction described in the previous section are respected. 

\begin{enumerate}
  \item \textbf{Compliance:} An agent will not utter a message that violates the constraints, or that makes it impossible to finish successfully the interaction in the steps determined by the bound (or in finite steps if there is no bound).
  \item \textbf{Non Abandonment:} An agent will not finish intentionally the conversation unless the constraints are fulfilled.
\end{enumerate} 

Agents learn when they decide how to interpret a received word. The overall method is simple: if $a_1$ receives $v_2$, it updates the value of $\omega(v_2,v_1)$ for all the words $v_1$ in a set $U \subseteq V_1$. The set $U$, unlike the set of possible interpretations $W$ used to decide the mapping, can have interpretations that are not possible messages in that particular moment, and the values in the interpretation distribution are updated according to that.

A first decision is how to choose the set $U$. Since checking satisfiability is computationally expensive, this can impact considerably the overall performance of the methods. A first option is to update the value of all possible words ($U = V_1$), which can be slow for large vocabularies. Another option is to use only the options that are considered until an interpretation is chosen, or all those words that have more value than the first possible interpretation. In this case, if $v_1$ is chosen as an interpretation for $v_2$, $U= \{ v \in V_1 \text{ such that } \omega(v_2, v) > \omega(v_2, v_1)\} $. A third possibility is a midpoint between these two, in which $U = \{ v \in V_1 \text{ such that } \omega(v_2, v) \geq \omega(v_2, v_1)\}$. This option updates also the words that have the same value as the chosen interpretation. If the distribution is uniform when the learning starts, this updates many possible interpretations in the beginning and fewer when there is more information. This is the approach we choose, and the one used for the evaluation.



We present two methods to update the values in the interpretation distribution. The first one is a general technique that only takes into account whether messages are possible. This method is not exclusively designed for protocols that use LTL, and can be used for constraints specified in any logic for which agents have a satisfiability procedure.  In a second method we show how the particular semantics of the protocols we introduced can be taken into account to improve the learning  process.

 
\subsection{Simple Strategy}

The first approach consists in updating the value of interpretations for a received message, according to whether it is coherent or not with the protocol and the performed interaction.
When a new message $v_2$ is received, the values for all interpretations are initialized as described in the last section. We use a simple learning method in which updates are a punishment or reward that is proportional to the value to update. We chose this method because it incorporates naturally previous alignments, it is very simple to formulate, and the proportional part updated can be manipulated easily to consider different situations, something that will be relevant for the second technique. However, it is not necessarily the most efficient choice, and other methods, such as one that records a history, can converge faster. Consider a reward and a punishment parameter $r_r, r_p \in [0,1]$. We use the notion of \emph{possible messages} in Definition \ref{def:pos} and update the values for $v_1 \in U$ as follows:

{\footnotesize
\begin{equation*}
\omega(v_2, v_1) : =
\begin{cases}
     \omega(v_2, v_1) + r_r \cdot \omega(v_2, v_1) & \mbox{if } \langle a_2, v_1 \rangle \mbox{ is possible } \vspace{0.2cm}\\ 
     \omega(v_2, v_1) - r_p \cdot \omega(v_2, v_1) & \mbox{if }  \langle a_2, v_1 \rangle \mbox{ otherwise } \ 
 
\end{cases}
\end{equation*}
}

After all the updates for a given message are made, the values are normalised in such a way that  $\sum_{v \in V_1} \omega(v_2, v) = 1$. Either simple sum-based or \textit{softmax} normalization can be used for this.


\subsection{Reasoning Strategy}

\begin{table*}[h!]
\begin{center}
\bgroup
\def\arraystretch{1.3}
\begin{tabular}{ |c | c | }\hline
  Violated Constraint & Action \\ \hline
  - & $\omega(v_2, v_1) : = \omega(v_2, v_1) - r_p \cdot \omega(v_2, v_1)$  \\\hline
  $absence(\langle a_2, v_1 \rangle,n)$ & $\omega(v_2, v_1) : = 0$  \\\hline
  \specialcell{\specialcell{$!correlation(\langle a_2, v_1 \rangle,\langle a_2, v'_1 \rangle$) \\ $!response(\langle a_2, v_1 \rangle,\langle a_2, v'_1 \rangle)$ } \\ \specialcell{$!before(\langle a_2, v_1 \rangle,\langle a_2, v'_1 \rangle)$ \\  \specialcell{ $!premise(\langle a_2, v_1 \rangle,\langle a_2, v'_1 \rangle)$ \\ $!imm\_after(\langle a_2, v_1 \rangle,\langle a_2, v'_1 \rangle)$}}}  &  \specialcell{$\omega(v_2, v_1) : = \omega(v_2, v_1) - r_p \cdot \omega(\mu(v'_1), v'_1)$ \\ $  \omega(\mu(v'_1), v'_1) : = \omega(\mu(v'_1), v'_1) - r_p \cdot \omega(v_2, v_1) $} \\\hline 

 \specialcell{$ premise(\langle a_2, v_1 \rangle,\langle a_2, v'_1 \rangle)$ \\$imm\_after(\langle a_2, v'_1 \rangle,\langle a_2, v_1 \rangle)$  } & \specialcell{$\omega(v_2, v_1) : = \omega(v_2, v_1) - r_p \cdot \omega(\mu(v'_1), v'_1)$\\$  \omega(\mu(v'_1), v'_1) : = \omega(\mu(v'_1), v'_1) - r_p \cdot \omega(v_2, v_1) $}\\\hline
 $before(v',v)$ &  $\omega(v_2, v_1) : = \omega(v_2, v_1) - r_p \cdot \omega(v_2, v_1) $ \\\hline 
 $relation(\langle a_1, v_1 \rangle,\langle a_2, v'_1 \rangle)$ &
$\omega(v'_1, v'_1) : = \omega(v'_1, v'_1) - r_p \cdot \omega(v_2, v_1) $ \\  \hline
\end{tabular}
\egroup
\end{center}
\caption{Updates for each violated constraint when a message is received}
\label{table:updateM}
\end{table*}

The simple strategy does not make use of some easily available information regarding which constraints were violated when a message is considered impossible. To illustrate the importance of this information, consider an interaction in which the Customer said \textsf{water} but the Waiter interpreted it as \textsf{vino}. If the Customer says \textsf{beer}, the Waiter may think that interpreting it as \textsf{birra} is impossible, since ordering two alcoholic beverages is not allowed ($!correlation(\mathsf{birra},\mathsf{vino})$ would be violated). However, this is only due to misunderstanding \textsf{water} in the first place. In the reasoning approach, agents make use of the semantics of violated rules to perform a more fine-grained update of the values. Before presenting the method, we need to divide constraints in two kinds, that will determine when agents can learn from them.

\begin{definition}
	A constraint $c$ is semantically \emph{non-monotonic} if there exist interactions $i$ and $i'$ such that $i$ is a prefix of $i'$, and $i \models c$ but $i' \not\models c$. A constraint $c$ is \emph{semantically monotonic} if this cannot happen, and $i \models c$ implies that also all its extensions are models of $c$.
\end{definition}

The following proposition divides the constraints in our protocols between monotonic and non-monotonic.

\begin{proposition}
  The constraints $existence$, $coexistence$ and $response$ are semantically monotonic. All the rest of the constraints defined are semantically non-monotonic.  
\end{proposition}

\begin{proof}
To prove non-monotonicity, it is enough to show an example that violates the constraint. For example, consider $before(m,m')$ and suppose an interaction $I$ such that $m \not\in i$ and $m' \not\in i$. Then $i \models before(m,m')$ and $i  \, . \,  m' \not\models before(m,m')$. Monotonicity can be proven by counterexample, showing that for a constraint $c$, if $i  \, . \,  m' \not\models c$ then necessarily $i \not\models c$. For example, if $i  \, . \,  m$ violates $existence(m', n)$, there are two options, either $m=m'$ or $m\neq m'$. Let $\#(m, i)$ and $\#(m, i  \, . \,  m')$ be the number of occurrences of $c$ in $i$ and $i  \, . \,  m'$ respective. In both cases, $\#(m, i  \, . \,  m') \geq \#(m, i)$, so if $\#(m, i  \, . \,  m') \leq n$, then necessarily $\#(m, i) \leq  n$.
\end{proof}

 Non-monotonic constraints can be used to learn while interacting, using the principle of compliance (if a constraint is violated, something must be wrong in the interpretations). Monotonic constraints could be used when the interaction ends, using the principle of non-abandonment (if not all constraints are satisfied, there must be an error). However, since our agents cannot communicate in any way, they ignore why the interaction ended, and therefore they do not know if their interlocutor considers the constraints were satisfied or not, making it very difficult to decide how to update values. In this work we focus on handling non-monotonic constraints, leaving monotonic ones for future work where an ending signal is introduced. To start, let us define formally when a non-monotonic constraint is violated.

\begin{definition}
  Given an interaction $i$ and an interpretation $v_1$ for a received message, a constraint $c$ is \emph{violated} if $i \models c$ but $i  \, . \,  \langle a_2, v_1 \rangle \not\models c$.
\end{definition}

If agent $a_1$ decides that $v_1$ is not a possible interpretation for $v_2$, let $Viol$ be all $c \in C_1$ that are violated by $v_1$. It is possible that $Viol = \emptyset$. This can happen when there are no violated constraints because a subset of constraints becomes unsatisfiable, for example, if the waiter has not said \textsf{size} and a protocol includes the rules $\{!response(\langle C, \mathsf{wine} \rangle,   \langle W, \mathsf{size} \rangle), $\\$  existence(\langle W, \mathsf{size} \rangle,1) \}$, the customer cannot say \textsf{wine} even when $Viol$ would be empty. Another situation in which a message can be impossible without any constraint being violated is when it makes the interaction impossible to finish successfully in the number of utterances indicated by the bound, or in finite time if there is no bound.

To explain the update in the reasoning strategy, consider again reward and a punishment parameter $r_r, r_p \in [0,1]$. If an interpretation $v_1$ is possible for a received message $v_2$, then $\omega(v_2, v_1) : = \omega(v_2, v_1) + r_r \cdot \omega(v_2, v_1)$. If it is not, it is updated as indicated in to Table \ref{table:updateM}, according to which rule was violated. These actions are performed iteratively for each broken rule. After the updates, all values are normalized to obtain $\sum_{v \in V_1} \omega(v_2, v) = 1$. In this case, we need to use a method that maintains the values that are $0$, so \textit{softmax} is not a possibility.

 Let us explain the motivation under each update. Violating the existential non-monotonic constraint is different to violating a relation constraint, because it expresses a condition over only one action, so if the chosen interpretation violates it, it must be wrong. For this reason, the value of the mapping is set to $0$.

Negation constraints express a relation over two actions, and if their are broken, necessarily both of them happened in the interaction, and if a constraint is violated, it could be because either of them was misinterpreted. To take this into account, we propose an update inspired in the Jeffrey's rule of conditioning \cite{Shafer1981Jeffrey}, that modifies the Bayes Rule of conditional probability for the cases when evidence is uncertain. If $Q$ is a probability distribution for a partition $E_1 \dots E_n$ of possible evidence, then Jeffrey's rule states that the posterior probability of $A$ is computed as:

\vspace{-0.3cm}
$$Q(A) = \sum_{i = 0}^{n} P(A|E_i) \cdot Q(E_i) $$

Back to our problem, let $v_2$ be a received message such that $\mu(v_2) = v_1$. Suppose $a_1$ receives a new message $v'_2$, and it discovers that mapping it with $v'_1$ violates a constraint $relation(v_1,v'_1)$. This means that $P(\alpha(v_2)= v_1 \land \alpha(v'_2)= v'_1) = 0$, and therefore $P(\alpha(v_2)= v_1 | \alpha(v'_2)= v'_1) = 0$. Then, if $Pos$ is the set of possible mappings, we can write:

\vspace{-0.3cm}
$$\omega(v_2, v_1) : = \omega(v_2, v_1) \cdot \sum_{v'_1 \in Pos} \omega(\alpha(v'_2, v'_1)$$

or, what is equivalent, if $Impos$ is the set of incompatible mappings, since $\sum_{v \in v_1} \omega(v_2, v_1) = 1$:

$$\omega(v_2, v_1) : = \omega(v_2, v_1) - \sum_{v'_1 \in Impos} \omega(v'_2, v'_1)$$
  
We implement this by subtracting for each mapping a value that is proportional to the confidence in the other mapping. Since agents do not explore all the possible interpretations, the value  is not computed exactly, but approximated considering only the interpretations that have more confidence.

When the interpretation depends on the interpretation of many other mappings, there are too many possibilities that difficult reasoning about which one can be wrong. In this case, the agents use a default punishment $r_p$. This occurs when the positive versions of $premise$ and $imm\_after$ are violated, when $Viol = \emptyset$, and when and a violated constraint depends on a message that $a_1$ sent, since they have no information about how their messages were interpreted by $v_2$.



Lastly, let us discuss the value of $r_p$. Choosing $r_p = \frac{1}{|V_1|}$, the technique has the effect of subtracting a larger value when the agent is confident in the correctness of the previous interpretation. This is the value that we use.

  



\section{Experimental Evaluation}
\label{sec:eval}
In this section we present the results of experiments that show how the methods we propose perform experimentally when used by agents with different vocabularies. Before presenting the results, let us discuss the generation of data for experimentation.

\begin{figure*}[h!]

  \begin{minipage}{.3\textwidth}
  \centering
   \includegraphics[scale=0.24]{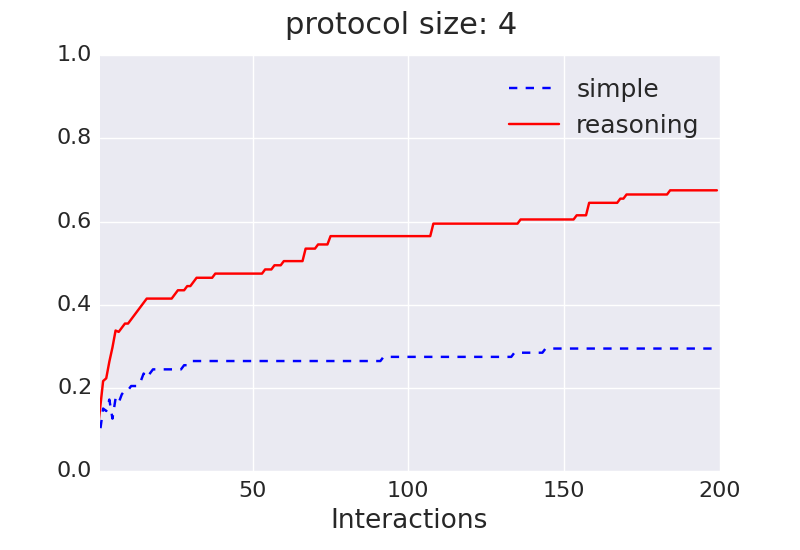}
   \label{fig:v10p06}
  \end{minipage} \quad
  \begin{minipage}{.3\textwidth}
  \centering
   \includegraphics[scale=0.24]{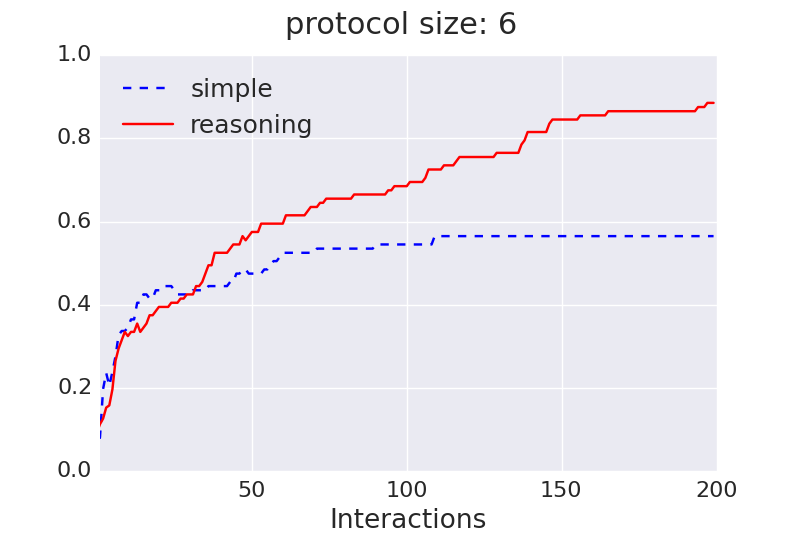}
   \label{fig:v10p06}
  \end{minipage} \quad
  \begin{minipage}{.3\textwidth}
 \includegraphics[scale=0.24]{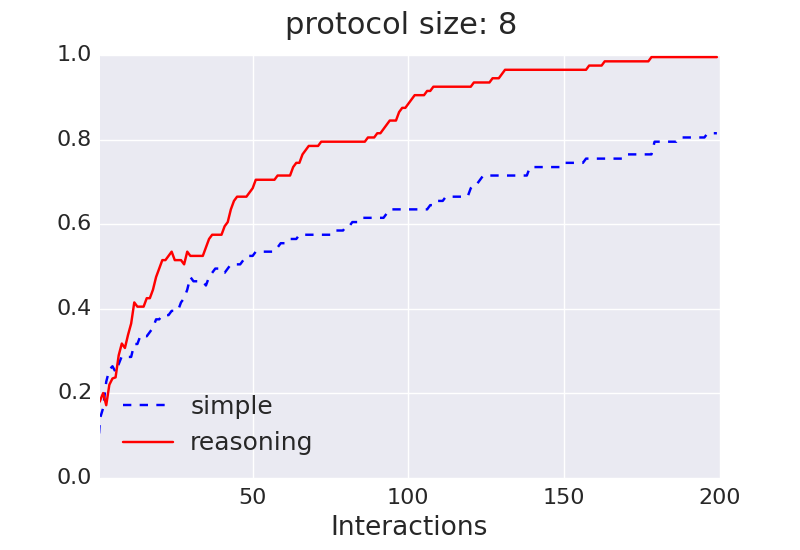}
   \label{fig:v10p08}
  \end{minipage}

\begin{minipage}{.6\textwidth}
  \centering
  \includegraphics[scale=0.24]{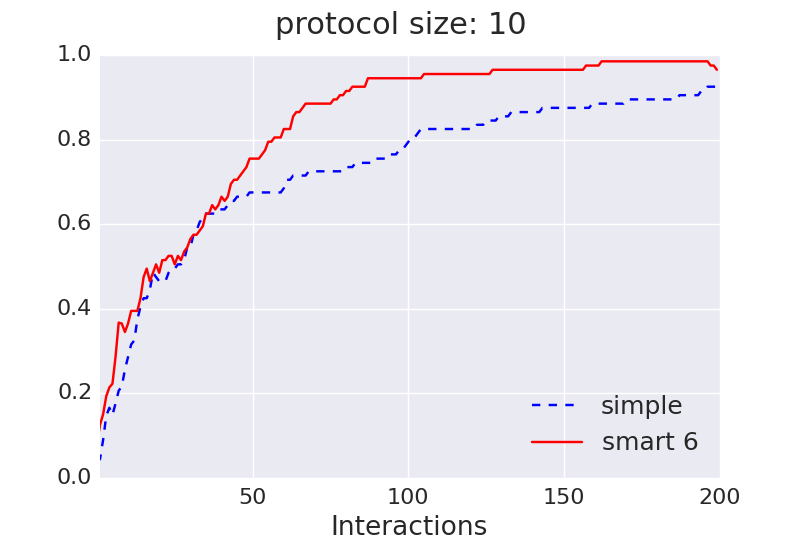}
   \label{fig:v10p10}
  \end{minipage} \begin{minipage}{.2\textwidth}
    \centering
    \vspace{0.05cm} \hspace{-2cm}
  \includegraphics[scale=0.24]{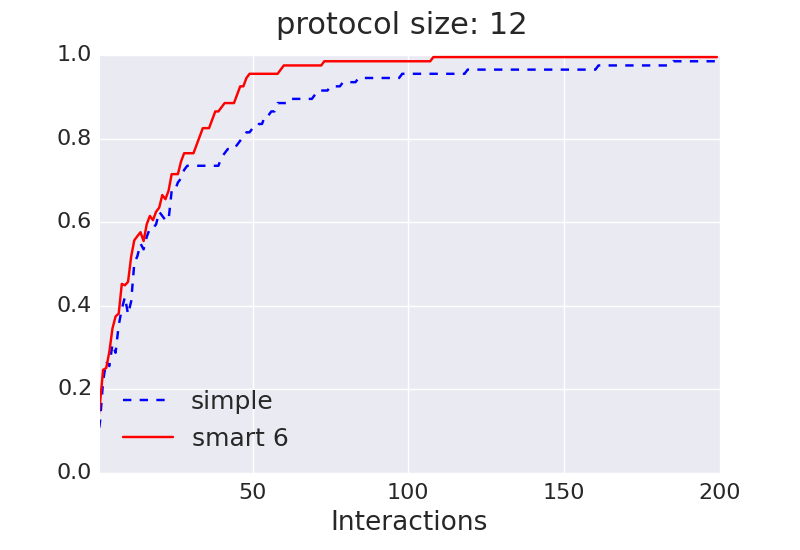}
   \label{fig:v10p12}
  \end{minipage}
 \caption{Results for a vocabulary of size 10}
 \label{fig:exp1}
\end{figure*}

Due to the lack of existing datasets, we performed all the experimentation with randomly generated data. We created protocols with different sizes of vocabulary and of set of constraints. To generate a protocol from a vocabulary, we randomly chose constraints and added them to the protocol if it remained satisfiable with the new rule. In the implementation, we express constraints over words instead of over messages, so a constraint are valid for any combination of senders. We used bounds in our experimentation to limit the possible constraints and to eventually finish the interactions between agents, that in all cases had the value of the vocabulary size. We created pairs of compatible protocols $\mathfrak{P_1},\mathfrak{P_2}$, with vocabularies $V_1$ and $V_2$ respectively by building a bijective alignment $\alpha : V_2 \rightarrow V_1$ and then using it to translate the constraints in a protocol $\mathfrak{P_2}$, obtaining $\mathfrak{P_1}$. We used the NuSMV model checker \cite{cimatti2002nusmv} to perform all the necessary satisfiability checks.

To evaluate how agents learn $\alpha$ from the experience of interacting, we first need to define a measure of how well an agent knows an alignment. Since agents always choose the possible mapping with highest weight, we can easily extract from the interpretation distribution an alignment that represents the first interpretation choice for each foreign word in the following interaction. This alignment, that we call $\alpha_{i}$ for agent $a_i$, is one of the mappings with highest weight for each foreign word. Formally, for $a_1$, the domain of $\alpha_1$ is $v_2 \in V_2$ such that $\omega(v_2,v)$ is defined for all $v \in V_1$, and $\alpha_1(v_2) = rnd(\argmax_{v \in V_1}\omega(v_2, v))$. Note that there are multiple possibilities with the same value, $\alpha_1(v_2)$ takes a random word between them.


To compare this alignment with $\alpha$, we used the standard precision and recall measures, and their harmonic mean combination, commonly known as \emph{F-score}. To use the standard definitions directly, we only need to consider alignments as relations instead of functions: given an alignment $\alpha: V_2 \rightarrow V_1$, the corresponding relation is the set of pairs $\langle v_2, v_1 \rangle$ such that  $\alpha(v_2) = v_1$.

\begin{definition}
\label{def:pir}
Given two alignments $\beta$ and $\gamma$ expressed as relations, the \emph{precision} of $\beta$ with respect to $\gamma$ is the fraction of the mappings in $\beta$ that are also in $\gamma$: $$\mathit{precision}(\beta,\gamma) = \dfrac{\mid \beta \cap \gamma \mid}{\mid \beta \mid}$$ while its \emph{recall} is the fraction of the mappings in $\gamma$ that were found by $\beta$: $$\mathit{recall}(\beta,\gamma) = \dfrac{\mid \beta \cap \gamma \mid}{\mid \gamma \mid}$$

Given an alignment $\beta$  and a reference alignment $\gamma$, the \emph{F-score} of the alignment is computed as follows: $$F-score(\beta,\gamma) = 2 \cdot \dfrac{\mathit{precision}(\beta,\gamma) \cdot \mathit{recall}(\beta,\gamma)}{\mathit{precision}(\beta,\gamma) + \mathit{recall}(\beta,\gamma)}$$
  \end{definition}

We performed experiments parametrized with a protocol and vocabulary size. In all experiments, agents are sequentially given pairs of compatible protocols, and after each interaction  we measure the F-score of their alignments with respect to $\alpha$. The same protocols can appear repeatedly in the sequence, but we consider that eventually new ones appear always in the interaction, which implies that at some point the meaning of all words is fixed. Our agents also used the NuSMV model checker for both checking satisfiability and finding violated constraints. We fist intended to perform experiments with the same vocabulary sizes that are used for testing in \cite{Atencia2012}, which are $5, 10, 15, 40, 80$. However, the interactions for size $80$ were too slow to be able to perform a reasonable amount of repetitions for each technique. This problem is intrinsic to the protocols we use (since agents have to decide if the messages they want to send are possible), and should be taken into account in future work. Since varying the vocabulary size did not provide particularly interesting results, we show here the experiments for a vocabulary of $10$ words and four different protocol sizes. We used punishment and rewards parameters of $r_p, r_r = 0.3$ for the simple strategy, which were best in a preliminary test. Each experiment was repeated 10 times. 




In a first experiment, we let agents interact 200 times and measured the way they learned the alignment for different protocol sizes, shown in Figure \ref{fig:exp1}. The \emph{F-score} if computed by averaging the \emph{F-score} of each agent that participates in the interaction. The curves show that the smart strategy is always better than the simple one, but this difference is more dramatic when the protocols are smaller. This is because there is less information, so using it intelligently makes a great difference. The curves are sharp in the beginning, when agents do not have much information about $\alpha$, and they make many mistakes from which they learn fast. When agents reach a reasonable level of precision and recall, mistakes are less frequent, and therefore the pace of the learning slows considerably. Although this affects the convergence to an F-score of $1.0$, it also means that, after a certain number of interactions, agents have learn enough to communicate successfully most of the times. The sparse mistakes that they will make in future conversations make them learn more slowly the remaining mappings. Table \ref{table:conv} shows how fast agents reach a reasonable level of F-score, that we chose of $0.8$ based on \cite{Euzenat2011}.


  

\begin{table}[h!]
\begin{center}
\bgroup
\def\arraystretch{1.3}
\begin{tabular}{ c c  c c c  }\hline
 & 6 & 8 & 10 & 12 \\ \hline
simple & - & 187 & 101 & 47  \\
smart & 139 & 87 & 57 & 33  \\
\end{tabular}
\egroup
\end{center}
\caption{0.8 convergence}
\label{table:conv}
\end{table}

In a second experiment, we studied the performance of our methods when used to repair existing alignments. To this aim, we created alignments with different values of precision and recall with respect to $\alpha$, that we divided into three categories. Low quality alignments had precision and recall 0.2, medium quality 0.5, and high quality 0.8. We only studied alignments with the same value of precision and recall to reduce the number of combinations  and have a clear division of quality levels. We made agents with this alignments interact, using the simple and the smart technique. The results can be seen in Figure \ref{fig:alg}. The number next to each agent class in the reference is the precision and recall of the alignment that was given to that agent. Again, the reasoning strategy performs particularly better when only little information is available; in this case, the repair is much better than when using the simple strategy for alignments of low quality.

\begin{figure}[h!]
\centering
  \includegraphics[scale=0.4]{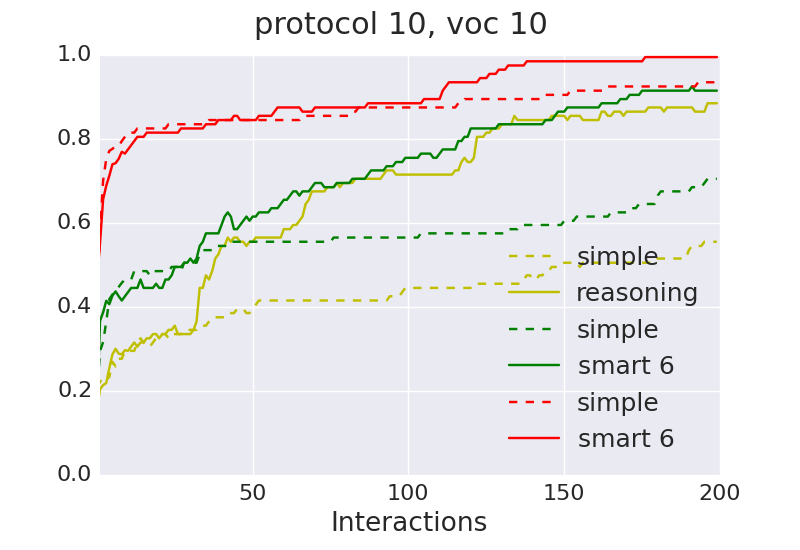}
 \caption{Results for different alignment qualities}
 \label{fig:alg}
\end{figure}

  

\section{Related Work}

As we already mentioned, a large portion of the existent work that tackles the problem of vocabulary alignment for multi-agent communication consists on techniques to let agents discuss and agree on a common vocabulary or alignment. Some of these approaches use argumentation or negotiation techniques \cite{Santos2016, laera2007argumentation}, performing an offline negotiation that takes place before actual interactions can start. Other approaches, such as the three-level schema developed by van Diggelen \cite{diggelen2007ontology}, discuss the meaning of works on a by-demand basis, that is, only when it is necessary to use them to perform some task. 

Other approaches use different kinds of grounding to learn alignments or the meaning of words from interactions. The well established research by Steels \cite{Steels1998} considers a situation in which learners share a physical environment, and have the ability of pointing to things to communicate if they are talking about the same things. Goldman et al.  \cite{Goldman2007}, investigated how agents can learn to communicate in a way that maximises rewards in an environment that can be modelled as a Markov Decision Process. In \cite{Barrett2014}, the authors study a version of the multi\-agent, multi\-armed bandit problem in which agents can communicate between each other with a common language, but message interpretations are not know. Closer to our approach is the work by Euzenat on cultural alignment repair \cite{Euzenat2014}, and particularly the already mentioned work by Atencia and Schorlemmer \cite{Atencia2012}, where the authors consider agents communicating with interaction protocols represented with finite state automata, and use a shared notion of task success as grounding. Our work, instead, considers only the coherence of the interaction as grounding. To the best of our knowledge, the work presented in this paper is the first approach in which agents learn languages dynamically from temporal interaction rules, taking into account only the coherence of the utterances. 

A very well studied problem consists on learning different structures from experience, such as grammars \cite{delaHiguera2010} or norms \cite{sen2007emergence}. Our approach can be seen as the reverse of these approaches, where some structure is considered shared and used to learn from.

\section{Conclusions and Future Work}

The techniques we propose allow agents to learn alignments between their vocabularies only by interacting, assuming that their protocols share the same set of models, without requiring any common meta-language or procedure. The assumption of sharing the models can be removed, but the learning will be slower because an extra component of uncertainty is added. Although our methods converge slowly, they use only very general information, and can be easily combined with other methods that provide more information about possible mappings, as we show with the integration of previous alignments. Our agents learn dynamically while interacting and the techniques achieve fast a reasonable level of knowledge of the alignment. Therefore, a possible way of using our techniques would be to go first through a training phase in which agents learn many mappings fast, and then start performing the real interactions, and keep learning to discover the remaining mappings. The technique that uses particular semantics of the protocols, as expected, improves the learning. These techniques can be particularly developed for each kind of protocols. Agents do not need to use the same techniques or even the same logic, as long as they share the models.


There exist many possible directions of research derived form this work. First, we could relax the assumption that agents do not share any meta-language, considering agents that can in some way exchange information about the interaction. For example, considering only that agents can communicate whether they finished a task successfully would make possible to reason about monotonic rules when an interaction ends. An approach like this would relate our work to the one by Santos \cite{Santos2016} about dialogues for meaning negotiation. Another possible extension consists in considering agents that prioritize the language learning to performing the task. In this case, agents would have to decide what utterances it is more convenient to make in order to get more information about the alignment.  An aspect that should be improved in future work is the performance in terms of runtime per interaction, since it was very slow for larger vocabularies.

\bibliographystyle{abbrv}
\bibliography{lib-corrected.bib}
\end{document}